\documentclass[11pt]{article}

\usepackage[margin=1in]{geometry}
\usepackage[utf8]{inputenc} % allow utf-8 input
\usepackage[T1]{fontenc}    % use 8-bit T1 fonts
\usepackage{hyperref}       % hyperlinks
\usepackage{url}            % simple URL typesetting
\usepackage{booktabs}       % professional-quality tables     % 
\usepackage{nicefrac}       % compact symbols for 1/2, etc.
\usepackage{microtype}      % microtypography
\usepackage{xcolor}
\usepackage{mathtools}
\usepackage{mathrsfs}
\usepackage[all]{xy}
\usepackage{mathpazo}
\usepackage{multirow}
\usepackage{amsmath,amssymb,amsthm,amsfonts,latexsym,bbm,xspace,graphicx,float,mathtools,dsfont}
\usepackage{makecell}
\usepackage{comment}
\usepackage{algorithm}
\usepackage{algorithmic}

\usepackage{accents}
\usepackage{appendix}
\usepackage{float}
\usepackage{hyperref}
\hypersetup{
	colorlinks=true,
	citecolor = blue       % false: boxed links; true: colored links
}

\newtheorem{definition}{Definition}[section]

\newtheorem{theorem}{Theorem}[section]
\newtheorem{proposition}{Proposition}[section]
\newtheorem{lemma}{Lemma}[section]

\newtheorem{claim}{Claim}[section]
\newtheorem{claim*}{Claim}[section]

\newtheorem{example}[theorem]{Example}

{\author{Johannes Brustle\footnote{
London School of Economics, London, UK. Email: \texttt{j.brustle@lse.ac.uk}}
 \and 
 Paul D\"utting \footnote{
Google Research, Zurich, Switzerland. Email: \texttt{duetting@google.com}}
\and
Balasubramanian Sivan
\footnote{
Google Research, New York City, New York, USA. Email: \texttt{balusivan@google.com}
}}}

\begin{document}

\title{Price Manipulability in First-Price Auctions}

\date{}

\newcommand{\NN}{\ensuremath{\mathbb{N}}}
\newcommand{\RR}{\ensuremath{\mathbb{R}}}
\newcommand{\EE}{\ensuremath{\mathbb{E}}}
\newcommand{\mE}{\mbox{$\mathbb E$}}

\newcommand{\bx}{\mathbf{x}}
\newcommand{\bp}{\mathbf{p}}
\newcommand{\bv}{\mathbf{v}}
\newcommand{\bb}{\mathbf{b}}
\newcommand{\bq}{\mathbf{q}}
\newcommand{\D}{\mathcal{D}}
\newcommand{\M}{\mathcal{M}}
\newcommand{\N}{\mathcal{N}}
\newcommand{\I}{\mathcal{I}}
\newcommand{\eps}{\mbox{$\epsilon $}}
\DeclarePairedDelimiter{\ceil}{\lceil}{\rceil}

\newcommand{\ps}{p^{\text{strategic}}}
\newcommand{\ph}{p^{\text{honest}}}
\newcommand{\winners}{W}

\newcommand{\Prr}[2]{\mbox{\rm\bf Pr}_{#1}\left[#2\right]}
\newcommand{\Ex}[2]{\mbox{\rm\bf E}_{#1}\left[#2\right]}

\newcommand{\OPT}{\mathrm{OPT}}
\newcommand{\ALG}{\mathrm{ALG}}
\newcommand{\erf}{\mathrm{erf}}
\newcommand{\Ber}{\mathrm{Ber}}
\newcommand{\growingmid}{\mathrel{}\middle|\mathrel{}}

\newcommand{\SOLD}{\mathrm{SOLD}}

%% Johannes' comments
%\newcommand{\johcomm}[1]{\textcolor{blue}{#1}}
\newcommand{\johcomm}[1]{\textcolor{black}{#1}}

%% disable color in paul's edits
\colorlet{teal}{black}
%%
%% Keywords. The author(s) should pick words that accurately describe
%% the work being presented. Separate the keywords with commas.
%\keywords{First-price auction, approximate incentive compatibility, strategyproofness in the large}

%%
%% This command processes the author and affiliation and title
%% information and builds the first part of the formatted document.
\maketitle
\begin{abstract}
 First-price auctions have many desirable properties, including uniquely possessing some, like credibility. %\textcolor{teal}{A potential drawback of first-price auctions is that they are inherently non-truthful and agents may have an incentive to strategize.} 
 %In repeated settings bidders can learn how to bid over time, but this can render the system instable and lead to inefficiencies. It is thus central to understand, when and to which extent there is room for manipulation.  
 %\textcolor{teal}{Their attractive properties have led to their adoption in ad exchanges for display ads.}
 \textcolor{teal}{However, first-price auctions are also inherently non-truthful, and non-truthfulness may result in instability and inefficiencies.} 
 %\textcolor{teal}{In fact, empirical evidence from the early days of sponsored search showed that the by-then industry standard of using first-price auctions led to sawtooth bidding patterns.}
 \textcolor{teal}{Given these pros and cons, we seek to quantify the extent to which first-price auctions are susceptible to manipulation.}
 
%For measuring proximity from truthfulness, 
\textcolor{teal}{In this work} we adopt a metric that was introduced in the context of bitcoin fee design markets: the percentage change in payment that can be achieved by being \textcolor{teal}{strategic}. \textcolor{teal}{We study the behavior of} this metric for single-unit and $k$-unit auction environments with $n$ i.i.d.~buyers, \textcolor{teal}{and seek conditions under which the percentage change tends to zero as $n$ grows large.}
% %, where each buyer participates in the auction independently with probability $p$. 
% For the single-unit case, we derive a sharp characterization of when the percentage change in payment goes to zero for any bounded-support distribution: if and only if \textcolor{teal}{$k = o(n)$}. 
% For \textcolor{teal}{monotone hazard rate (MHR)} distributions, we show the same is true even for unbounded support distributions, and also show that MHR is necessary to be able to remove the bounded support restriction: for the class of  $\alpha$-strongly regular distributions for any $\alpha < 1$ \textcolor{teal}{($\alpha = 1$ is MHR)}, the boundedness is necessary to get percentage change in payment going to zero. For the $k$-unit case, we again derive a tight characterization: the percentage change in payment for any bounded-support distribution goes to zero if and only if \textcolor{teal}{$k = o(n)$}. \textcolor{teal}{We don't fully resolve the unbounded distribution case, but show that interestingly even for $k$ as small as $\sqrt{n}$ there are MHR distributions for which the price manipulability does \emph{not} vanish.}
% \textcolor{teal}{All our results also extend to the case where participation in the auction is probabilistic, and bidders participate in the auction independently with probability $p$.}

\textcolor{teal}{To the best of our knowledge, ours is the first rigorous study of the extent to which large multi-unit first price auctions are susceptible to manipulation. We provide an almost complete picture of the conditions under which they are ``truthful in the large,'' and exhibit some surprising boundaries.} 
%\textcolor{teal}{An interesting avenue for future work, apart from settling the exact boundary for unbounded MHR distributions, is to examine the non-asymptotic behavior.}
\end{abstract}
\section{Introduction}\label{sec:intro}
First-price auctions enjoy many desirable properties including transparency (winners-pay-bid)  \textcolor{teal}{\cite{Krishna09}}, and uniquely possess some properties like credibility~\cite{AL20}. They have been put to use in a large variety of settings. Their attractive properties have been influential in causing the recent switch in the display ads industry's ad exchanges to first-price auctions. 
\textcolor{teal}{A potential drawback of first-price auctions is that they are not truthful. That is, bidders may benefit from misreporting their valuation/willingness to pay. Non-truthfulness is problematic for a number of reasons. A main concern is that in the absence of a dominant-strategy equilibrium, it's usually a non-trivial task to find an equilibrium bidding strategy. Moreover, if bidders have to learn how to bid then they will update their bids frequently, causing a lot of traffic to the website, and resulting in an instable and upredictable system as a whole.}

%% The paper is:
%% https://web.stanford.edu/~ost/papers/cycling.pdf

\textcolor{teal}{A cautionary tale regarding the use of first-price auctions comes from the early days of search ads \cite{Edelman07}. In their influential paper, Edelman and Ostrovsky performed an analysis of bidding data from Overture, which was one of the first companies to provide a paid placement service, and used a first-price auction to determine the allocation of ads and payments. Their analysis showed that bids were very unstable and often exhibited a sawtooth pattern of gradual price increases followed by sudden drops (cf.~Figure 2 in their paper).}

\textcolor{teal}{Given these pros and cons of first-price auctions, it seems in-dispensable---and of great practical relevance---to have tools that allow to quantify the extent to which first price-auctions are susceptible to manipulation.}
%, or put positively, how close to being incentive compatible they are.}

\subsection{Our Approach}

%%In this work we ask, \emph{exactly how truthful a first-price auction is}. 

%We study $k$-units first-price auctions for both $k=\Theta(1)$ and $k =o(n)$. The case of $k=\Theta(n)$ is trivial and can be easily shown to be susceptible to manipulation. We include it just for completeness. In a $k$-units first-price auctions, the highest $k$ bidders receive the good, and each of them pays their bid. The rest pay $0$. There are $n$ buyers, and their private values are drawn i.i.d. from some distribution with cdf $F$.

\textcolor{teal}{We adopt a notion of approximate incentive compatibility that  was recently introduced by \cite{LSZ19}, who used it to study different bitcoin fee designs.} The idea is to measure the maximum relative change in payment that any agent can achieve by being strategic. \textcolor{teal}{We use this metric to study the $k$-unit first-price auction with $n$ bidders, whose values are i.i.d.~draws from $F$. This auction receives one bid from each agent. The $k$ highest bidding agents ``win'' and are assigned one unit of the good, and have to pay their bid.}

\textcolor{teal}{To formally define the metric, for} any given valuation profile $\bv$, let $v_{(r)}$ denote the $r$-th largest value. 
%And let $\winners(\bv) = \{i | v_i \geq v_{(k)}\}$. 
Let $\ph_i(\bv) \equiv v_i$ whenever $v_i \geq v_{(k)}$ and $0$ otherwise, denote the agent $i$'s payment when all the agents, including $i$, bid their true value. Let $\ps_i(\bv) \equiv v_{(k+1)}$ whenever $v_i \geq v_{(k)}$ and $0$ otherwise, denote the smallest payment the agent $i$ can manage to get charged and still get allocated, when $i$ is strategic and all other agents are honest. The relative price change for agent $i$ is given by $\delta_i(\bv, n, k) \equiv 1 - \ps_i(\bv)/\ph_i(\bv)$ whenever $v_i \geq \bv_{(k)}$ and $0$ otherwise. The quantity of interest for us is 
\begin{align}\EE[\delta_{\max}(\bv, n, k)]  \equiv \EE[\max_i \delta_i(\bv, n, k)],\label{eq:delta-max}
\end{align}
namely, the maximum relative price change that can occur for any agent. For brevity, from here on we drop the arguments of $\delta_{\max}$. 

\begin{example}[Two uniform bidders] \textcolor{teal}{Suppose that there are two bidders with $v_i$ for $i \in \{1,2\}$ drawn independently from $U[0,1]$, and that there is one item.}
\textcolor{teal}{Consider the random variables $v_{\max} = \max\{v_1,v_2\}$ and $v_{\min} = \min\{v_1,v_2\}$. Then $\ph_i(\bv) = v_{\max}$ if $v_i = v_{\max}$ and $\ph_i(\bv) = 0$ otherwise. Similarly, $\ps_i(\bv) = v_{\min}$ if $v_i = v_{\max}$ and  $\ps_i(\bv) = 0$ otherwise.}
\textcolor{teal}{The joint density of $(v_{\min},v_{\max})$ is $f(v_{\min},v_{\max}) = 2$. So \[\EE[v_{\min}/v_{\max}] = \int_{0}^{1} \int_{v_{\min}}^{1} \frac{v_{\min}}{v_{\max}} \cdot 2 \;dv_{\max} \;dv_{\min} = \frac{1}{2}\]
and $\EE[\delta_{\max}] = 1-\EE[v_{\min}/v_{\max}] =  1/2$.}
\end{example}

In this work we ask the following: 
%\begin{enumerate}
%\item 
1) 
%When is a first-price auction almost truthful, i.e., when does $\EE[\delta_{\max}] \rightarrow$ 0? 
\johcomm{When do the benefits from bid manipulation vanish to zero in the large market limit, i.e., $\EE[\delta_{\max}] \rightarrow$ 0?}
%\item 
2) Which distributions allow $\EE[\delta_{\max}] \rightarrow 0$? 
%\item 
3) Do the distributions need to be bounded? Does regularity help? Does monotone hazard rate (MHR) help? 
%\end{enumerate}
The answers are not straightforward, and are \emph{not} as simple as letting $n \to \infty$ will automatically let $\EE[\delta_{\max}] \rightarrow 0$. 

\subsection{Discussion}

%\textcolor{teal}{[discuss metric: why is it desirable and different from other metrics?  (1) applicability no use of valuations (2) complete info (3) in all settings we consider asymptotically BNE = truthful, and in definitions above we could also ask for price manipuabilit with respect to equilibrium bids, we will see that BNE is not sufficient for no price manipuability  ]}

\textcolor{teal}{\noindent 1) \emph{Deviation from truthfulness or from BNE?} The definition in Equation~(\ref{eq:delta-max}) measures the maximum percentage price change \emph{assuming everyone else bids truthfully}. Wouldn't it be more natural to define it with respect to some Bayes-Nash equilibrium (BNE)? It turns out that for all settings we consider, that is whenever $k = o(n)$ and $n \rightarrow \infty$, then a) there is a unique BNE \cite{ChawlaH13} (this part is true for all $n$ and $k$) and b) a simple application of Myerson's payment identity \cite{Myerson81} shows that asymptotically the two quantities coincide.}

\medskip

\noindent\textcolor{teal}{2) \emph{Why this measure and not something else?} A first reason is that it's very practitioner friendly. Indeed, in sharp contrast to other metrics that have been proposed for measuring approximate incentive compatibility (see Section~\ref{sec:related}), the metric defined in Equation~(\ref{eq:delta-max}) does \emph{not} require the auctioneer to reverse engineer the agents' valuations to e.g.~draw conclusions about the utility they could gain from a misreport. It only requires knowledge of the bids and prices, which is precisely the data the auctioneer has access to. A second reason is that the notion of ICness captured by Equation~(\ref{eq:delta-max}) is ``ex post'' (or ``full info''). It measures the extent to which an agent that has learned about the other agents' bids can game the system.} 
%%
%A few desirable properties of this metric are: a) when it is close to zero, it is evidence that the auction is close to being truthful, b) when it is bounded away from zero, it is evidence\footnote{\johcomm{This is a quite common way of formalizing “truthfulness in the large”. For example, see definition 4 in \cite{Budish18} - they measure utility while we measure price changes, but other than that the concept of convergence to zero is the same.}} that there is room to manipulate/shade bids, and c)
%it is concrete and simple to compute: 
%\johcomm{it is concrete, simple to compute and does not need knowledge of true valuations:}
%if $X^n_{(r)}$ denotes the $r$-th order statistic among $n$ i.i.d. draws of a random variable, it follows that $\EE[\delta_{\max}(\bv,n,k)] = 1 - \EE\left[\frac{X^n_{(k+1)}}{X^n_{(1)}}\right]$.
\johcomm{As observed in \cite{Edelman07}{ (Section 4)}, there is indeed evidence that bidders observe their competitors' recent behavior and use this information to adjust their own bids.} 
%This justifies assuming complete information of all other bids. e) In this definition of the metric, any agent not winning the item is truthful. However, we later see that if these agents bid according to a BNE, asymptotically, the analysis remains the same. One conclusion is that the BNE assumption does not always suffice to guarantee no price manipulability.

\medskip

\noindent\textcolor{teal}{3) \emph{Expected ratio vs. ratio of expectations.}}
\johcomm{It is important to note that the main part of our analysis consists of studying the expectation of ratios of order statistics, which is non-trivial and quite different from the ratio of expectations of order statistics.} \textcolor{teal}{Indeed, if $X^n_{(r)}$ denotes the $r$-th order statistic among $n$ i.i.d. draws of a random variable, it follows that $\EE[\delta_{\max}(\bv,n,k)] = 1 - \EE[X^n_{(k+1)}/X^n_{(1)}].$}

\subsection{Our Results}

%\paragraph{Results and surprises} 
Our results are summarized in Table~\ref{tab:results_n}. %The table is self-explanatory: 
For different classes of distributions, for different regimes of $k$ covering the whole spectrum from $k=1\dots n$, we ask whether every distribution in that class provides no room for bid manipulation (i.e., whether $\lim_{n\to\infty}\EE[\delta_{\max}] = 0$), or, whether there are distributions in that class permitting bid manipulation (i.e., whether $\lim_{n\to\infty}\EE[\delta_{\max}] > 0$). 

\begin{table*}
  \centering
  \begin{tabular}{lllp{3.4cm}}
    \toprule
     & Bounded support & MHR, continuous &\parbox{5cm}{$\alpha$-strongly regular,\\  $\alpha \in [0,1)$}\\
    \midrule
    $k = \Theta(1)$  & $\lim_{n\to\infty}\EE[\delta_{\max}] = 0$  & $\lim_{n\to\infty}\EE[\delta_{\max}] = 0$  & $\lim_{n\to\infty}\EE[\delta_{\max}] > 0 $  \\
    $k = o(n)$      & $\lim_{n\to\infty}\EE[\delta_{\max}] = 0$ & $\lim_{n\to\infty}\EE[\delta_{\max}] > 0$  & $\lim_{n\to\infty}\EE[\delta_{\max}] > 0$\\
    \midrule
    $k = \Theta(n)$     & $\lim_{n\to\infty}\EE[\delta_{\max}] > 0$ & $\lim_{n\to\infty}\EE[\delta_{\max}] > 0$  & $\lim_{n\to\infty}\EE[\delta_{\max}] > 0$\\
    \bottomrule
  \end{tabular}
  \caption{Relative price change for first-price auctions. The last row is immediate, and has been included for completeness.}
  \label{tab:results_n}
\end{table*}

\textcolor{teal}{For the single-unit case, we derive a sharp characterization of when the percentage change in payment goes to zero. For any bounded-support distribution the percentage change in payment goes to zero. For unbounded support distributions, we show that monotone hazard rate (MHR) is a \emph{necessary and sufficient condition} for percentage change in payment to vanish. In particular, for the class of  $\alpha$-strongly regular distributions, for any $\alpha < 1$ \textcolor{teal}{($\alpha = 1$ is MHR)}, the percentage change in payment does not vanish for unbounded support distributions. For the $k$-unit case, we again derive a tight characterization: the percentage change in payment for any bounded-support distribution goes to zero if and only if \textcolor{teal}{$k = o(n)$}. \textcolor{teal}{We don't fully resolve the unbounded distribution case, but show that, interestingly, for $k$ as small as $\sqrt{n}$ the percentage change in payment may \emph{not} vanish. How small should $k$ be for it to vanish is a concrete open question.}}

\textcolor{teal}{Obtaining \textcolor{teal}{such sharp separation results} entails finding the precise property that makes a distribution susceptible or robust to bid manipulation \textcolor{teal}{(which we develop in Section~\ref{sec:unbounded})}. We identify this quantity to be the limit of the ratio of the median-of-the-maximum-of-$n$-draws and median-of-the-maximum-of-$2n$-draws as $n\to\infty$. Roughly, we show that distributions where this ratio is strictly lesser than $1$ are susceptible to bid manipulation, while distributions where this ratio is $1$ (for not just $2n$ draws in the denominator, but $\ell n$ draws for any constant $\ell$), are robust to bid manipulation. The novelty is in arriving at this insight. \emph{The proofs are simple and clean to read after this key insight is found}.}

\textcolor{teal}{We find our results somewhat surprising in two ways: (1) Usually, the class of MHR distributions\footnote{Intuitively MHR distributions are to be thought of as distributions with tails no heavier than the exponential distribution. See Section~\ref{sec:prelim} for a formal definition.} behave similarly, at least qualitatively, to the class of $\alpha$-strongly regular distributions\footnote{See Section~\ref{sec:prelim} for formal definitions of regular and $\alpha$-strongly regular distributions.}~\cite{CR14}. Indeed, a similar question to ours was asked of a different mechanism: the guarded empirical revenue maximization algorithm (guarded ERM) in~\cite{DLLQWY20}. They show that guarded ERM is incentive compatible as $n\to\infty$ for $\alpha$-strongly regular distributions for \emph{every} $\alpha > 0$. In stark contrast, in this paper, we establish a sharp separation between MHR ($1$-strongly regular) and $\alpha$-strongly regular distributions for every $\alpha < 1$, \emph{even for the $1$-unit auction, i.e., $k=1$}. (2) Moreover, contrary to what one may suspect, we show that even for MHR distributions $k = o(n)$ is \emph{not} sufficient for bid manipuability to vanish. Rather, even within the class of MHR distributions there is a reasonably sharp separation across different $k$.}

\paragraph{Probabilistic participation.}

\textcolor{teal}{All our results actually hold for a more general setting, in which not all the $n$ agents need to participate. In this more general model, each of the $n$ agents participates independently with a probability $p$.}
%For $k=\Theta(1)$, for bounded support distributions, as long as $np = \omega(1)$ (a generalization of $n=\omega(1)$ for the simpler setting) all our results hold. For larger $k$, our results hold as long as $k = o(np)$ (a generalization of the $k=o(n)$ for the simpler setting). 

The probabilistic participation is practically motivated by quota-based throttling that happens in the display ads industry: many buyers cannot adequately handle the extremely high Queries-Per-Second (QPS) that an ad exchange sends them, and hence ask for the requests sent to them to be throttled. The latter is often done probabilistically. 

\subsection{Related Work}\label{sec:related}

\paragraph{Relative price change metric.}

The relative price change metric was introduced in the context of the bitcoin fee design market by~\cite{LSZ19}. The goal here was to analyze the monopolistic price (MP) mechanism. The MP mechanism solicits $n$ bids $b_1,\dots,b_n$ and sets a price $p$ that maximizes $k\cdot b_{(k)}$. To what extent is there a temptation for users to shade their bids in the MP mechanism was the question they studied. They showed that for finite support distributions, $\lim_{n \to \infty} \EE[\delta_{\max}] = 0$ for the MP mechanism, and conjectured that it should hold for all bounded support distributions. 
Later,~\cite{Yao20} confirmed the conjecture for all bounded support distributions in $[1,D]$, and also show that for regular distributions $\lim_{n \to \infty} \EE[\delta_{\max}] > 0$. More recently,~\cite{DLLQWY20} study the guarded\footnote{``Guarded'' just forces the monopolistic mechanism to pick a price that will let at least a certain minimum number of bids win} empirical revenue maximization mechanism, and show that even for $\alpha$-strongly regular distributions for any $\alpha > 0$, $\lim_{n \to \infty} \EE[\delta_{\max}] = 0$. 
%%In a different line of work,~\cite{Giannakopoulos21} consider anonymous posted price auctions for the setting of $n$ bidders with i.i.d values for one item. In order to obtain their results on MHR and strongly regular distributions, the authors also need to first establish statements involving tail bounds of order statistics. The technical challenges however differ from ours since they are concerned with the expectation of a single order statistic whereas we need to understand the fraction of the $k$-th over the largest order statistic.
%To the best of our knowledge, we are the first to systematically study the first-price auctions with the same metric, and we cover the entire spectrum of $k$, and also $\alpha$'s in $\alpha$-strongly regular distribution. The results are reasonably different from those for the MP mechanism, including the separation between MHR and $\alpha$-strongly regular, the nuanced separation between $k=\Theta(1)$ and $k=\sqrt{n}$ even for MHR distributions etc.

\paragraph{Other approaches to approximate IC.} \textcolor{teal}{There is a number of alternative approaches to measuring approximate incentive compatibility. A popular one is based on the notion of \emph{regret} \cite{ParkesKE01,DayM08,DuttingFJLLP12,DuttingNPR19,BalcanSV19,FengSS19,Colini-Baldeschi20}, which measures how much an agent can improve her utility (= value minus payment) through a misreport. Different papers differ in how they aggregate regret across types and agents.}

\textcolor{teal}{A particularly relevant paper that uses a regret-based approach comes from \cite{AzevB18} who propose the concept of strategy-proofness in the large (SP-L). This notion requires that it is
approximately optimal for agents to report their types truthfully in sufficiently large markets.}

\textcolor{teal}{Another related line of work is on testing incentive compatibility \cite{LahaieEtAl18,DengL19}. Here the basic question is how agents participating in a repeated auction can test whether the auction is incentive compatible. Their approach is based on segmenting query traffic into buckets and systematically perturbing bids.}

\textcolor{teal}{A similar idea was used recently \cite{DengMLZ20} to propose a IC metric based on Myerson's payment identity. The notion compares the expected utility by uniform bid shading or bid inflation by a factor $(1+\alpha)$. It evaluates to $1$ if the auction is IC, if it is smaller or larger than that, then an agent may benefit from increasing/decreasing its bid uniformly.}

\textcolor{teal}{Other papers (such as \cite{PathakS13,TroyanM20}) propose metrics of approximate IC that are applicable to settings where the utilities are ordinal (and their is no notion of money).} %\cite{PathakS13}, for example, proposes to rank mechanisms based on the number of instances that
%are manipulable. Another direction is to define ``obvious manipulations'' and classify non-IC mechanisms into those that are obviously manipulable and those that are not \cite{TroyanM20}.}

\subsection{Open Questions}

\textcolor{teal}{A first obvious open question is finding the exact boundary for (unbounded) MHR distributions with $k > 1$ units. We show that price manipulability does not vanish for $k \geq \sqrt{n}$ (so $k = o(n)$ is not sufficient), but how small does $k$ have to be for price manipulability to vanish?}

A \textcolor{teal}{second} general direction is to study feasibility constraints beyond $k$-units, such as matroids and beyond, along with probabilistic participation. A good place to start would be the $G(n,p)$ model where each edge in the graph participates independently with probability $p$. What can we show about the relative price change in these settings?

\textcolor{teal}{It would also be interesting to study the non-asymptotic behavior of the metric we proposed, e.g., in form of upper bounds on the percentage price change and how fast it decreases to zero as a function of $n, k$ and properties of the distribution.}

%\textcolor{teal}{Finally, it would be interesting to explore other non-truthful auction formats, and to compare different non-truthful auction formats to each other.}

\section{Preliminaries}
\label{sec:prelim}
\paragraph{Setting, first-price auction.} We consider a setting with $n$ agents, with values drawn i.i.d. from a distribution with cdf $F$. There are $k$ identical units of an item to be sold, and each agent is unit-demand. The seller runs a first-price auction: solicit bids, award the $k$ units to the $k$ highest bids, and each allocated buyer pays their bid. We consider a setting with probabilistic participation, where each agents participates independently with probability $p$. Thus the number of participating agents $m~\sim B(n,p)$ is binomially distributed with parameters $n,p$.

\paragraph{Regular, $\alpha$-regular distributions, MHR.} Regular distributions are those for which $\phi(v) = v - (1-F(v))/f(v)$ is a weakly increasing function of $v$. An $\alpha$-regular distribution is one for which $\phi(v') - \phi(v) \geq \alpha(v' - v)$ for all $v',v$ in the support. MHR distributions are those for which $f(v)/(1-F(v))$ is a weakly increasing function of $v$. Equivalently, a $0$-regular distribution is a regular distribution and a $1$-regular distribution is a MHR distribution. In the literature they are referred to as $\alpha$-strongly regular distributions, and we refer to them as $\alpha$-regular for short. 

Our results for bounded support distributions hold for all distributions: discrete, continuous, mixed. Our results for MHR distributions assume that the distributions are continuous. Our results for $\alpha$-strongly regular distributions are negative results, and hence continuity doesn't matter.

\paragraph{Relative price change.} For any given valuation profile $\bv$, let $v_{(r)}$ denote the $r$-th largest value. 
Let $\ph_i(\bv) \equiv v_i$ whenever $v_i \geq v_{(k)}$ and $0$ otherwise, denote the agent $i$'s payment when all the agents, including $i$, bid their true value. Let $\ps_i(\bv) \equiv v_{(k+1)}$ whenever $v_i \geq v_{(k)}$ and $0$ otherwise, denote the smallest payment the agent $i$ can manage to get charged and still get allocated, when $i$ is strategic and all other agents are honest. The relative price change for agent $i$ is given by $\delta_i(\bv, n, k) \equiv 1 - \ps_i(\bv)/\ph_i(\bv)$ whenever $v_i \geq \bv_{(k)}$ and $0$ otherwise. The quantity of interest for us is $\EE[\delta_{\max}(\bv, n, p, k)] \equiv \EE_{m\sim B(n,p), \bv \sim F^m}[\max_i \delta_i(\bv, m, k)]$, namely, the maximum relative price change that can occur for any agent, when $m\sim B(n,p)$ agents participate. For brevity, from here on we often drop the arguments of $\delta_{\max}$. 

\paragraph{Order statistics.} If $X^n_{(r)}$ denotes the $r$-th order statistic among $n$ i.i.d. draws of a random variable, it follows that 
$$\EE[\delta_{\max}(\bv,n,p, k)] = 1 - \EE_{m\sim B(n,p)}\left[\frac{X^m_{(k+1)}}{X^m_{(1)}}\right].$$
For convenience, we define $X_{(j)}^i =0$ for $j>i$.

\subsection{General approach and binomial tail bounds}\label{sec:binomial}
We split the task of lower bounding $\EE\left[\frac{X^m_{(k+1)}}{X^m_{(1)}}\right]$ into two subtasks using the following inequality:
%As discussed in Section~\ref{sec:prelim} we are interested in studying $1 - \EE_{m\sim B(n,p)}\left[\frac{X^m_{(k+1)}}{X^m_{(1)}}\right]$. In this work, we lower bound $\EE\left[\frac{X^m_{(k+1)}}{X^m_{(1)}}\right]$ using 

%Consider the first-price auction to sell $k$ units of a good, to a set $\N$ of $n$ agents (here $n\geq k$). Values of agents are drawn i.i.d from some distribution $\D$. 
%The quantity we study is $\EE_{v \sim \D}[\delta_{\max}(v)]$ (we drop the subscript under $v$ as it is clear):

%$$ \EE[\delta_{\max}(v)] = 1-\EE\left[\frac{X_{(k+1)}^n}{X_{(1)}^n}\right]$$

%Define $X_{(j)}^i =0$ for $j>i$. Given this, we can let $i \in \N$ appear i.i.d with probability $p$. Let $m\leq n$ be the remaining number of agents, then $m \sim B(n,p)$, the binomial distribution. We will study the properties of
%$$ \EE[\delta_{\max}(v)] = 1-\EE\left[\frac{X_{(k+1)}^m}{X_{(1)}^m}\right]$$ as $n\rightarrow \infty$.

\begin{equation}\label{firstineq}
    \mE\left[\frac{X_{(k+1)}^m}{X_{(1)}^m}\right] \geq \Pr[m>i] \cdot \EE\left[\frac{X_{(k+1)}^i}{X_{(1)}^i}\right].
\end{equation}
%Although the two factors on the right hand side of (\ref{firstineq}) will be treated separately, in using this inequality it is important to note that for tight results in terms of $k$, it is necessary to obtain tight results for $\Pr[m>i]$ in terms of $i$. We use the next part of this section to do the latter.
To effectively use this inequality, we obtain tight binomial tail bounds, i.e., bounds on $\Pr[m>i]$. 
%\subsection{Binomial Bounds}\label{sec:binomial_bounds}
%In this section we prove the following lemma.
\begin{lemma}\label{binomial_bound}
For $m\sim B(n,p)$, and for any $i = o(np)$, $\lim_{n\rightarrow \infty}\Pr[m > i]=1$.
\end{lemma}

While the proof of Lemma~\ref{binomial_bound} is relatively straightforward, for completeness, we give a proof in Appendix~\ref{app:binomial_bound}. From the next section onwards, the focus is on analyzing $\EE\left[\frac{X_{(k+1)}^i}{X_{(1)}^i}\right]$.

\section{Warmup: Bounded support distributions}\label{sec:bounded}
%In the case of bounded distributions, we can fully classify the convergence of $\EE[\delta_{\max}(v)]$. 
It is easy to see that for $k=\Theta(np)$, $\lim_{n\rightarrow \infty}\EE[\delta_{\max}(v)]>0$. To see this, consider the Uniform $[0,1]$ distribution and $p=1$: it is easy to see that $\lim_{n \rightarrow \infty}\mE \left[\frac{X_{(k+1)}^n}{X_{(1)}^n} \right]<1$. For every other value of $k$, we obtain a positive result: namely, whenever $k = o(np)$, $\lim_{n\rightarrow \infty}\EE[\delta_{\max}(v)] = 0$, or equivalently $\lim_{n \rightarrow \infty}\mE \left[\frac{X_{(k+1)}^n}{X_{(1)}^n} \right]=1$.
\begin{lemma}\label{bounded_dist}\label{lem:bounded}
For any bounded support distribution $\D$, i.e., $\Pr_{v\sim \D}[v \in [0,H]] =1$ for some $H < \infty$, and $k=o(n)$, we have 
\[
\lim_{n \rightarrow \infty}\mE \left[\frac{X_{(k+1)}^n}{X_{(1)}^n} \right]=1.
\]
\end{lemma}

% \textcolor{teal}{The proof of Lemma~\ref{lem:bounded}, which appears in Appendix~\ref{app:bounded}, applies a Chernoff bound argument.}

\begin{proof}
The end of support $H$ can always be chosen such that $\Pr_{v\sim \D}[v \in [0,H]] =1$, and, $\forall \delta>0$, $\exists \epsilon > 0$ such that $\Pr_{v\sim\D}[v  \geq H-\delta] = \epsilon$.

Given the $n$ i.i.d. draws $X_1\ldots,X_n$ be $n$ from $\D$, let $Z_n =1$ if $X_n \geq H-\delta$ and $Z_n=0$ otherwise. 
%By the weak law of large numbers, for all $\omega>0$, $
%\Pr\left[\left| \frac{\sum_{j=1}^nZ_j}{n}-\epsilon\right|>\omega\right] \rightarrow 0$.
%For instance let $\omega=\frac{\epsilon}{2}$. 
By Chernoff bounds, it follows that as $n\to\infty$, the probability $\Pr[\text {At least } \frac{\epsilon n}{2} \text{ out of } X_1,\ldots,X_n \text{ are at least} \\
%% WWW HERE
H-\delta] \rightarrow 1$. Here $\frac{\epsilon n}{2}$ is an arbitrary choice, and $\epsilon n (1-\epsilon')$ for any constant $0 < \epsilon'<1$ works.

%In other words, given any $\delta>0$, let $\epsilon$ be as defined above. Since $k=o(n)$, for $n$ large enough we have
%$$\Pr[X_{(k+1)}^n \geq H-\delta] \geq 
%\Pr[X_{(\frac{\epsilon}{2} \cdot n)}^n\geq H-\delta] \rightarrow %1.
%$$
Hence for any $\delta>0$,
\begin{align*}
\mE \left[\frac{X_{(k+1)}^n}{X_{(1)}^n} \right] &\geq \frac{H-\delta}{H}\Pr[X_{(k+1)}^n \geq H-\delta] 
\geq \frac{H-\delta}{H}\Pr[X_{(\frac{\epsilon n}{2})}^n \geq H-\delta] \rightarrow \frac{H-\delta}{H}.
\end{align*}
As $\delta \to 0$, we have $\frac{H-\delta}{H} \to 1$, completing the proof.
\end{proof}

\begin{theorem}
%Let the First Price auction be given by the $k$-uniform matroid on $n$ agents each participating i.i.d with probability $p = \frac{\omega(1)}{n}$ and $k=o(np)$. If values are drawn i.i.d from some bounded distribution $\D$, then
%$$ \lim_{n\rightarrow \infty}\EE[\delta_{\max}(v)] =0.$$
For any bounded support distribution $\D$, and $k=o(np)$, $ \lim_{n\rightarrow \infty}\EE[\delta_{\max}] =0$.
\end{theorem}
\begin{proof}
Combine inequality (\ref{firstineq}) with Lemmas~\ref{binomial_bound} and~\ref{bounded_dist}. In other words, for any $k=o(np)$ we can always find an $i$ s.t. $k = o(i)$ and $i = o(np)$. Use that $i$ in Lemma~\ref{binomial_bound}, and use that $i$ instead of $n$ in Lemma~\ref{bounded_dist}.
\end{proof}
\section{Median and its relation to manipulation susceptibility}\label{sec:unbounded}
%In this section, we assume values of agents are given by random variables $X_1,\ldots,X_n$ drawn i.i.d from a continuous probability distribution $\D$ supported on $\RR_{\geq 0}$. We consider the case $k=1$.
In this section \textcolor{teal}{we develop our main technical insights. We} focus on the case of $k=1$ for clarity of presentation, although the results readily carry over to $k=\Theta(1)$.

\begin{definition}[Median of maximum]\label{def:tn}
Let $t_r \in \RR_{\geq 0}$ be a median of the maximum of $r$ draws from the distribution $\D$, defined as, the smallest real number s.t.  
$\Pr[X_1,\ldots,X_r \leq t_r] = \frac{1}{2}.$
\end{definition}

%It will be clear which distribution $\{t_n\}_{n\in \NN}$ corresponds to.
%\subsection{Distributions susceptible to bid shading}
We now establish in Lemma~\ref{find_counterexamples} a sufficient condition for bid manipulation using $t_r$. Later in Lemma~\ref{tn-condition} we establish that almost the same condition, with a minor strengthening, is also a necessary condition.
\begin{lemma}\label{find_counterexamples}
If $\D$ is a distribution for which $\lim_{n\rightarrow \infty}\frac{t_n}{t_{2n}}<1$, then 
$\lim_{n \rightarrow \infty} \mE\left[\frac{X_{(2)}^n}{X_{(1)}^n} \right]< 1$.
\end{lemma}

\begin{proof}
By definition of $t_n$,
$$\Pr[X_1, \ldots, X_n \leq t_n] = \Pr[X_1 \lor  \ldots  \lor X_n > t_n] = \frac{1}{2}.$$
Therefore,
\begin{align}
\Pr[X_1,\ldots,X_{2n} \leq t_{n}]
&= \Pr[(X_1,\ldots,X_{n} \leq t_{n}) \land (X_{n + 1},\ldots, X_{2n} \leq t_{n})] \nonumber\\ 
\label{2n_inequality}
&=(\Pr[X_1,\ldots,X_{n} \leq t_{n}])^2 = \frac{1}{4} 
\end{align}

Let $A_n$ be the event that exactly one of $X_1,\ldots,X_{2n}$ is strictly greater than $t_{2n}$ and all other random variables at at most $t_n$. By definition of $A_n$, it follows that $\Pr\left[ \left[\frac{X_{(2)}^{2n}}{X_{(1)}^{2n}} \right] \leq \frac{t_{n}}{t_{2n}} \right] \geq \Pr[A_n]$. Let $X_{-i} = \{X_1,\ldots,X_{2n}\} \setminus X_i$. We now lower bound the probability of event $A_n$ using inequality~\ref{2n_inequality} as follows.
\begin{align}
\label{eq:const_prob_median}
\Pr\left[ \left[\frac{X_{(2)}^{2n}}{X_{(1)}^{2n}} \right] \leq \frac{t_{n}}{t_{2n}} \right] &\geq \Pr[A_n] \nonumber\\
&=
\sum_{i=1}^{2n}\Pr[X_{-i} \leq t_{n} | X_i > t_{2n}]\Pr[X_i > t_{2n}] \nonumber\\
&= \sum_{i=1}^{2n}\Pr[X_{-i} \leq t_{n}]\Pr[X_i > t_{2n}] \nonumber\\
&\geq \frac{1}{4} \sum_{i=1}^{2n} \Pr[X_i > t_{2n}] \nonumber \\
&\geq \frac{1}{4} \Pr[X_{(1)}^{2n} > t_{2n}] =\frac{1}{8}.
%&\Rightarrow \Pr\left[ \left[\frac{X_{(2)}^{2n}}{X_{(1)}^{2n}} \right] \leq \frac{t_{n}}{t_{2n}} \right] \geq \frac{1}{8}.
\end{align}

So if the distribution $\D$ is such that $\lim_{n \rightarrow \infty} \frac{t_{n}}{t_{2n}} < 1$, then inequality~\eqref{eq:const_prob_median} shows that with a constant probability the ratio of second- and first-order statistics is strictly smaller than 1. Thus, the expectation of the ratio is also strictly smaller than $1$, i.e., $\lim_{n \rightarrow \infty} \mE\left[\frac{X_{(2)}^{2n}}{X_{(1)}^{2n}} \right]< 1$. This is equivalent to $\lim_{n \rightarrow \infty} \mE\left[\frac{X_{(2)}^{n}}{X_{(1)}^{n}} \right]< 1$.
\end{proof}

We use Lemma~\ref{find_counterexamples} to show $\lim_{n\rightarrow \infty} \EE[\delta_{\max}(v)] > 0$ for many distributions, including $\alpha$-strongly regular distributions for every $\alpha < 1$. We do so in Section~\ref{sec:alpha}. 
%\subsection{Distributions with asymptotically Incentive Compatible behaviour}

\textcolor{teal}{We now move on} to proving an almost-converse of Lemma~\ref{find_counterexamples}, namely, Lemma~\ref{tn-condition}. The proof of Lemma~\ref{tn-condition} crucially uses Claim~\ref{char_deltamax} which we state and prove first.
%That is, $1$-strongly regular distributions, also called MHR distributions, are asymptotically incentive compatible in the sense that $\lim_{n\rightarrow \infty} \EE[\delta_{\max}] =0$.
%In order to prove the result, we prove Lemma~\ref{tn-condition}, which can be thought of as the converse of Lemma \ref{find_counterexamples}.

\begin{claim}\label{char_deltamax}
If $\lim_{n\rightarrow \infty}\mE \left[\frac{X_{(2)}^n}{X_{(1)}^n} \right] \neq 1$, $\exists$ a subsequence $(a_{n(t)},b_{n(t)})$ with $a_{n(t)} < b_{n(t)}$ and constants $c>0$, $d<1$, such that, $\forall t \in \NN$  we have
$\Pr[X_{(2)}^{n(t)}<a_{n(t)} \wedge X_{(1)}^{n(t)}>b_{n(t)}] \geq c$ and $\frac{a_{n(t)}}{b_{n(t)}}\leq d$.
\end{claim}
\begin{proof}
For any real-valued non-negative random variables $X$ and $Y$ satisfying $\EE\left[\frac{X}{Y}\right] <1$, there must exist positive real numbers $a,b$ such that $a < b$ and $\Pr[X<a \ \land \ Y>b]>0$. 
By definition of limit, since
$\lim_{n\rightarrow \infty}\mE \left[\frac{X_{(2)}^n}{X_{(1)}^n} \right] \neq 1$ and $\mE \left[\frac{X_{(2)}^n}{X_{(1)}^n} \right] \in [0,1]$ $\forall n\in \NN$,
there exists some ${n(t)}_{\{t\in \NN\}}$, subsequence of $\NN$ and $\epsilon>0$ such that for all $t\in \NN$, $\EE \left[\frac{X^{n(t)}_{(2)}}{X^{n(t)}_{(1)}} \right]\leq 1-\epsilon$. Now assume by contradiction to the claim statement that for all corresponding sequences of upper and lower bounds $\{(a_{n(t)},b_{n(t)})\}$ satisfying $a_{n(t)} < b_{n(t)}$, either $\lim_{t\rightarrow \infty} \Pr[X^{n(t)}_{(2)}<a_{n(t)} \land X^{n(t)}_{(1)}>b_{n(t)}] = 0$ or $\lim_{t\rightarrow \infty} \frac{a_{n(t)}}{b_{n(t)}} \rightarrow 1$. In either case, this would imply $\lim_{t\rightarrow \infty}\EE \left[\frac{X^{n(t)}_{(2)}}{X^{n(t)}_{(1)}} \right]>1-\epsilon$, which is a contradiction to what we just established.
\end{proof}

%\begin{claim}\label{char_deltamax2}
%If $\lim_{n\rightarrow \infty}\mE \left[\frac{X_{(2)}^n}{X_{(1)}^n} \right]<1$, $\exists$ some $n_0 \in \NN$, a sequence $(a_{n},b_{n})$ with $a_n < b_n$ and constants $c>0$, $d<1$, such that, $\forall n>n_0$  we have
%$\Pr[X_{(2)}^{n}<a_{n} \wedge X_{(1)}^{n}>b_{n}] \geq c$ and %$\frac{a_{n}}{b_{n}}\leq d$.
%\end{claim}
%\begin{proof}
%For any real-valued non-negative random variables $X$ and $Y$ satisfying $\EE\left[\frac{X}{Y}\right] <1$, there must exist positive real numbers $a,b$ such that $a < b$ and $\Pr[X<a \ \land \ Y>b]>0$. Now consider a sequence of such random variables with the same property in the limit, that is $\lim_{n\rightarrow \infty}\EE \left[\frac{X_n}{Y_n} \right]<1$.
%By definition of limit, there exists some $n_0$ and $\epsilon>0$ such that for all $n>n_0$, $\EE \left[\frac{X_{n}}{Y_{n}} \right]\leq 1-\epsilon$. Now assume by contradiction to the claim statement that for all sequences $\{(a_n,b_n)\}$ satisfying $a_n < b_n$, either $\lim_{n\rightarrow \infty} \frac{a_n}{b_n} \rightarrow 1$ or $\Pr[X<a_n \land Y>b_n] \rightarrow 0$. In either case, this would imply $\lim_{n\rightarrow \infty}\EE \left[\frac{X_n}{Y_n} \right]>1-\epsilon$, which is a contradiction to what we just established.
%\end{proof}

\begin{lemma}\label{tn-condition}
For i.i.d random variables $X_1,\ldots,X_n$ drawn from some distribution, if for any constant $l\in \NN$ we have $\lim_{n \rightarrow 
\infty}\frac{t_{n}}{t_{nl}} =1$, then 
$\lim_{n\rightarrow \infty}\mE \left[\frac{X_{(2)}^n}{X_{(1)}^n} \right]=1$.
\end{lemma}
\begin{proof}
We show that whenever for any constant $l\in \NN$, \textcolor{teal}{it holds that} $\lim_{n \rightarrow 
\infty}\frac{t_{n}}{t_{nl}} =1$, a sequence as described in Claim~\ref{char_deltamax} does not exist. Thus, by the contra-positive of Claim~\ref{char_deltamax} the lemma is proved.

To show this, assume by contradiction that a sequence as in Claim~\ref{char_deltamax} exists, i.e., there exists $(a_n,b_n)$, $c>0$, $d<1$ as in Claim~\ref{char_deltamax}. Let $l$ be a constant s.t.\\
$\Pr[X_{(2)}^{nl} \leq t_{n}]<c$ and $\Pr[X_{(1)}^{nl}>t_{nl^2}]<c$. Such $l$ exists since
\begin{align}
\label{eq:nl2}
\Pr[X_{(2)}^{nl}\leq t_{n}] &\leq
\Pr[(X_1,\ldots,X_{\frac{nl}{2}}) \leq t_{n} \lor (X_{\frac{nl}{2}+1},\ldots,X_{nl}) \leq t_{n}]\nonumber\\
&\leq 2\Pr[(X_1,\ldots,X_{\frac{nl}{2}}) \leq t_{n}] \nonumber\\
&=  (\frac{1}{2})^{\frac{l}{2}-1} \text{ (by applying Definition~\ref{def:tn})}.
\end{align}
Similarly, by Definition~\ref{def:tn},
$\Pr[X_{(1)}^{nl^2}>t_{nl^2}]=\frac{1}{2}$, and likewise
$\Pr[(X_1,\ldots,X_{nl^2}) \leq t_{nl^2}] = \frac{1}{2}$ 
\begin{align}
\label{eq:nl1}
&\Rightarrow \Pr[(X_1,\ldots,X_{nl}) \leq t_{nl^2}] = (\frac{1}{2})^{\frac{1}{l}}\nonumber\\
&\Rightarrow \Pr[X_{(1)}^{nl}>t_{nl^2}] = 1-(\frac{1}{2})^{\frac{1}{l}}.
\end{align}
As the RHS of both inequalities~\eqref{eq:nl1} and~\eqref{eq:nl2} go to $0$ for large $l$, the only way the first part of our contrary assumption can be true, namely that $\Pr[X_{(2)}^{nl}<a_{n} \wedge X_{(1)}^{nl}>b_{n}] \geq c$ can be true, is that 
$a_n > t_{n}$ and $b_n < t_{nl^2}$. But when $n\to\infty$, this is impossible because, by the other part of our contrary assumption we have $\frac{a_{n}}{b_{n}} \leq d $, whereas
$\lim_{n\to\infty} \frac{t_{n}}{t_{nl^2}} = 1$ by what is given in the lemma statement.
\end{proof}

We use Lemma~\ref{tn-condition} to show $\lim_{n\rightarrow \infty} \EE[\delta_{\max}(v)] = 0$ for all MHR distributions. We do so in Section~\ref{sec:MHR}. 

\section{MHR and $\alpha$-strongly regular distributions}\label{sec:MHRRegular}

\johcomm{In this section we make use of the insights developed in Section~\ref{sec:unbounded} to analyze the behavior of $\lim_{n\rightarrow \infty} \EE[\delta_{\max}(v)]$ for MHR and $\alpha$-strongly regular distributions. We observe that for $k=1$,  \textcolor{teal}{or more generally $k = \Theta(1)$}, MHR suffices to get no price manipulability while strong regularity does not. However, this guarantee no longer holds for MHR when we increase $k$ from constant to $\sqrt{n}$.}

\subsection{$\alpha$-strongly regular distributions}\label{sec:alpha}
\textcolor{teal}{First we} use Lemma~\ref{find_counterexamples} to show $\lim_{n\rightarrow \infty} \EE[\delta_{\max}(v)] > 0$ for $\alpha$-strongly regular distributions for every $\alpha < 1$. 
\begin{theorem}\label{lem:alpha_regular}
For every $\alpha<1$, there exists an $\alpha$-strongly regular distribution for which, with $p=1$ and $k=1$,
$\lim_{n\rightarrow \infty} \EE[\delta_{\max}(v)] > 0$.
\end{theorem}
\begin{proof}
Let distribution $\D$ be given by $\Pr_{v\sim D}[v\leq x] = 1-\frac{1}{x^d}$ for some $d \in \NN$. Such a $\D$ is $(1-\frac{1}{d})$-strongly regular.
We compute $\lim_{n \rightarrow \infty} \frac{t_{n}}{t_{2n}}$ and show that it is less than $1$. Applying Lemma~\ref{find_counterexamples}, this immediately completes the proof.

By the definition of $t_n$ (see Definition~\ref{def:tn}):

\begin{align*}
\Pr[X_1,\ldots,X_n<t_n] = \frac{1}{2}
&\Rightarrow \left(1-\frac{1}{(t_n)^d}\right)^n = \frac{1}{2} \\
&\Rightarrow 1-\left(\frac{1}{2}\right)^{\frac{1}{n}} = \frac{1}{(t_n)^d}\\
&\Rightarrow t_n = \frac{1}{\left( 1-\left(\frac{1}{2}\right)^{\frac{1}{n}}\right)^{\frac{1}{d}}}\\
&\Rightarrow \frac{t_{n}}{t_{2n}}
=\left(\frac{\left( 1-\left(\frac{1}{2}\right)^{\frac{1}{2n}}\right)}{\left( 1-\left(\frac{1}{2}\right)^{\frac{1}{n}}\right)}\right)^{\frac{1}{d}} 
\end{align*}
By substituting $z=\left(\frac{1}{2}\right)^{\frac{1}{n}}$:
\begin{align*}
\lim_{n\rightarrow \infty}\frac{t_n}{t_{2n}} &= \lim_{n\rightarrow \infty} \left(\frac{\left( 1-\left(\frac{1}{2}\right)^{\frac{1}{2n}}\right)}{\left( 1-\left(\frac{1}{2}\right)^{\frac{1}{n}}\right)}\right)^{\frac{1}{d}} \\
&=\lim_{z\rightarrow 1^-}\left(\frac{\left( 1-z^{\frac{1}{2}}\right)}{\left( 1-z\right)}\right)^{\frac{1}{d}} \\
& =\lim_{z\rightarrow 1^-}\left(\frac{1}{\left( 1+z^{\frac{1}{2}}\right)}\right)^{\frac{1}{d}}=\left(\frac{1}{2}\right)^\frac{1}{d} < 1.
\end{align*}
Since $p=1$, we have:
$$
\lim_{n\rightarrow \infty} \EE[\delta_{\max}]  = 1-
\lim_{n \rightarrow \infty} \mE\left[\frac{X_{(2)}^n}{X_{(1)}^n} \right].
$$
As discussed at the beginning of the proof, the result follows from applying Lemma~\ref{find_counterexamples}.
\end{proof}

\subsection{MHR distributions}\label{sec:MHR}
%An interesting application for which Lemma \ref{tn-condition} is relevant is the exponential distribution.
\textcolor{teal}{Next we} use Lemma~\ref{tn-condition} to show $\lim_{n\rightarrow \infty} \EE[\delta_{\max}(v)] = 0$ for MHR distributions. 
\textcolor{teal}{As in Section~\ref{sec:unbounded} we focus on the case $k = 1$, but note that the arguments readily extend to $k = \Theta(1)$.}

%First as a warmup, we show this for the exponential distribution in Lemma~\ref{tn_exponential}, before showing it for all MHR distributions in Theorem~\ref{thm:MHR}. \textcolor{red}{Paul: Can we make this a lemma about exponential distributions, which states the properties which we use later?} \johcomm{Johannes: Tried to fix this.}
\textcolor{teal}{Our results build on the following lemma about exponential distributions, which are known to be extremal for the class of MHR distributions.}

\begin{lemma}\label{tn_exponential}
Consider the exponential distribution $\D$, with c.d.f $F(x)=1-e^{-cx}$ for a constant $c$.
For  $r\in (0,1)$, defining $s_n$ by \\
$\Pr[X_1,\ldots,X_n<s_n] = r$
implies 
$$\frac{s_n}{s_m} = \frac{\ln(1-r^{\frac{1}{n}})}{\ln(1-r^{\frac{1}{m}})}$$

Furthermore, for $l\in \NN$,
$$
\lim_{n \rightarrow \infty} \frac{\ln(1-r^{\frac{1}{n}})}{\ln(1-r^{\frac{1}{nl}})}=1.
$$
In particular, $$\lim_{n \rightarrow \infty} \frac{t_{n}}{t_{nl}} = 1.$$
\end{lemma}

\begin{proof}
For the first part, 
\begin{align*}
&\Pr[X_1,\ldots,X_n<s_n] = r \Rightarrow
(\Pr[X_1<s_n])^n = r\\
&\Rightarrow 
1-e^{-cs_n} = r^{\frac{1}{n}}\\
&\Rightarrow 
s_n = \frac{-\ln(1-r^{\frac{1}{n}})}{c}
\Rightarrow \frac{s_n}{s_m} = \frac{\ln(1-r^{\frac{1}{n}})}{\ln(1-r^{\frac{1}{m}})}
\end{align*}

For $m=nl$, this limit can be computed.

Substitute $z=r^{\frac{1}{nl}}$.
$$
\lim_{n \rightarrow \infty} \frac{\ln(1-r^{\frac{1}{n}})}{\ln(1-r^{\frac{1}{nl}})}
= \lim_{z\rightarrow 1^-}\frac{\ln(1-z^l)}{\ln(1-z)}
$$
By applying l'Hopital's rule twice, we get
$$
= \lim_{z\rightarrow 1^-}\frac{-lz^{l-1}}{1-z^l}\cdot\frac{1-z}{-1}
= \lim_{z\rightarrow 1^-} \frac{l(z^{l-1}-z^l)}{1-z^l}
$$

$$
= \lim_{z\rightarrow 1^-}\frac{l((l-1)z^{l-2} - lz^{l-1})}{-lz^{l-1}} = \lim_{z\rightarrow 1^-} \frac{lz-(l-1)}{z} = 1.
$$
\end{proof}
% We now generalize the argument to all MHR distributions.

\begin{theorem}\label{thm:MHR}
Let $\D$ be any MHR distribution and $k=1$. Then for $np = \omega(1)$,
$\lim_{n\rightarrow \infty} \EE[\delta_{\max}] =0$.
\end{theorem}
\begin{proof}
Let $F(x)$ be the c.d.f of the distribution. Since $\D$ is MHR, we know $\xi(x) = \ln(1-F(x))$ is concave ~\cite{Barlow64}. 
%For simplicity of presentation assume $\xi$ is everywhere differentiable. 
Moreover, since $F(x)$ is monotone increasing, $F(0)=0$ and $\lim_{x \rightarrow \infty}F(x)=1$, we get $\xi(0)=0$ and $\lim_{x\rightarrow \infty}\xi(x)=-\infty$.
Thus $\xi^{-1}(y)$ is well defined on $y\in (-\infty,0]$, and is non-negative concave monotone decreasing.
Thus $F^{-1}(y) = \xi^{-1}(\ln(1-y))$. We check that the conditions for Lemma~\ref{tn-condition} hold. We begin by using Definition~\ref{def:tn} that $t_n = F^{-1}((\frac{1}{2})^{\frac{1}{n}})$.
$$
\frac{t_{n}}{t_{nl}} = 
\frac{F^{-1}((\frac{1}{2})^{\frac{1}{n}})}{F^{-1}((\frac{1}{2})^{\frac{1}{nl}})} = \frac{\xi^{-1}(\ln(1-(\frac{1}{2})^{\frac{1}{n}}))}{\xi^{-1}(\ln(1-(\frac{1}{2})^{\frac{1}{nl}}))}
$$
By Lemma \ref{tn_exponential}, 
%\textcolor{red}{Paul: This is not the statement of Lemma 5.1} \johcomm{Johannes: fixed?}
\begin{equation}\label{exp_tn}
    \lim_{n\rightarrow \infty}\frac{\ln(1-(\frac{1}{2})^{\frac{1}{n}})}{\ln(1-(\frac{1}{2})^{\frac{1}{nl}})}=1
\end{equation}

We now show that a version of~\eqref{exp_tn} with $\xi^{-1}$ in the numerator and denominator is also true, thus proving the theorem. Denote $a(x) = \ln(1-(\frac{1}{2})^{\frac{1}{x}})$ and $b(x) = \ln(1-(\frac{1}{2})^{\frac{1}{xl}})$.
Assume for contradiction that there exists some sequence $\{x_t\}_{t\in \NN} \rightarrow \infty$ such that $\frac{\xi^{-1}(a(x_t))}{\xi^{-1}(b(x_t))}<d<1$ for some constant $d$.
$$
\Leftrightarrow (\frac{1}{d}-1)\xi^{-1}(a(x_t))<\xi^{-1}(b(x_t))-\xi^{-1}(a(x_t))
$$
Denote by $z(x_t)$ a subgradient of $\xi^{-1}$ at $a(x_t)$. By concavity of $\xi^{-1}$ we get
$\xi^{-1}(a(x_t)) \geq z(x_t)a(x_t)$ 
and
$\xi^{-1}(b(x_t))-\xi^{-1}(a(x_t)) \leq z(x_t)(b(x_t)-a(x_t))$. Thus 
\begin{align*}
(\frac{1}{d}-1)z(x_t)a(x_t)&\leq (\frac{1}{d}-1)\xi^{-1}(a(x_t))\\
&<\xi^{-1}(b(x_t))-\xi^{-1}(a(x_t))\\
&\leq z(x_t)(b(x_t)-a(x_t))
\end{align*}

We may divide by $z(x_t)$ on both sides since it is $<0$ for all $x_t>0$.
$$
\Leftrightarrow
(\frac{1}{d}-1) \leq \frac{b(x_t)-a(x_t)}{a(x_t)}
$$
This contradicts (\ref{exp_tn}), or
$$
\lim_{t\rightarrow \infty}\frac{b(x_t)}{a(x_t)}=1.
$$
Hence the conditions for Lemma~\ref{tn-condition} are verified, implying
$$\lim_{n\rightarrow \infty}\mE \left[\frac{X_{(2)}^n}{X_{(1)}^n} \right]=1.$$
As for bounded distributions, the result then follows from (\ref{firstineq}) and Corollary \ref{binomial_bound}.
\end{proof}

It is important to note that unlike for bounded support distributions, even for $p=1$, the positive result for $k=1$ unbounded support MHR distributions does not generalize to $k=o(n)$. As discussed in the introduction, even for small $k$ like $k=\sqrt{n}$ we show that MHR distributions leave room for manipulation, showing a sharp separation even within the class of MHR distributions.

\begin{proposition}
In a $k$-units first-price auction with $k=\big\lceil{\sqrt{n}}\big\rceil$ and $p=1$, the exponential distribution (an MHR distribution) has $\lim_{n\rightarrow \infty}\EE[\delta_{\max}(v)]>0.$
\end{proposition}
\begin{proof}
Let $q_r$ be defined by $\Pr[X_1,\ldots,X_r \leq q_r] = \frac{9}{10}$.
Hence $\Pr[X_1\lor\ldots \lor X_n > q_n] = \frac{1}{10}$, or equivalently $\Pr[X^n_{(1)} > q_n] = \frac{1}{10}$.\\
We now show that $\Pr[X^{n^2}_{(n)} > q_n] \leq\frac{1}{2}$ by method of contradiction. Assume by contradiction that
$\Pr[X^{n^2}_{(n)} > q_n]>\frac{1}{2}$. Then 
\begin{align*}
\Pr[X^n_{(1)} > q_n] 
&= \Pr[\text{uniform at random n-subset of } X_1,\ldots,X_{n^2} \text{ contains at}\\&\phantom{xxxxx}\text{least 1 element }>q_n]\\
&\geq \Pr[\text{ At least } n \text{ of } X_1,\ldots,X_{n^2}  \text{ are } >q_n]\\
&\qquad \times \Pr[\text{ at least one that is } >q_n \text{ is chosen in our random}\\&\phantom{xxxxxxxxx}\text{$n$-subset}\ | \ n \text{ out of } n^2 \text{ variables are } >q_n]\\
&\geq \frac{1}{2}(1-(1-\frac{1}{n})^n)
\end{align*}
As $n\to\infty$, the above quantity approaches $\frac{1}{2}(1-\frac{1}{e}) > \frac{1}{10}$, which is a contradiction to the fact that $\Pr[X^n_{(1)} > q_n] = \frac{1}{10}$.

Thus $\Pr[X^{n^2}_{(n)} > q_n]\leq \frac{1}{2}$
or equivalently $\Pr[X^{n^2}_{(n)}\leq q_n]\geq \frac{1}{2}$.
This is enough for a constant lower bound on the event we're interested in.
\begin{align*}
\Pr[X_{(n)}^{n^2} \leq q_n \land X_{(1)}^{n^2} > q_{n^2}] &=\Pr[X_{(n)}^{n^2} \leq q_{n} \ | \ X_{(1)}^{n^2} > q_{n^2}]\Pr[X_{(1)}^{n^2} > q_{n^2}]\\
&=\frac{1}{10}\Pr[X_{(n)}^{n^2} \leq q_{n} \ | \ X_{(1)}^{n^2} > q_{n^2}]\\
&\geq  \frac{1}{10}\Pr[X_{(n)}^{n^2} \leq q_{n} \wedge X_{(2)}^{n^2} \leq q_{n^2}\ | \ X_{(1)}^{n^2} > q_{n^2}]\\
&= \frac{1}{10}\Pr[X_{(n)}^{n^2} \leq q_{n} \ | \ X_{(2)}^{n^2} \leq q_{n^2}, \  X_{(1)}^{n^2} > q_{n^2}]\Pr[X_{(2)}^{n^2} \leq q_{n^2}]\\
&\geq \frac{9}{100}\Pr[X_{(n)}^{n^2} \leq q_{n} \ | \ X_{(2)}^{n^2} \leq q_{n^2}] \\
&\geq\frac{9}{100}\Pr[X_{(n)}^{n^2} \leq q_{n}]\geq \frac{9}{200}.
\end{align*}
From previous calculations in Lemma \ref{tn_exponential}, we know that for the exponential distribution %\textcolor{red}{Paul: Again, not the statement of Lemma 5.1} \johcomm{Johannes: fixed?}

$$
\frac{q_{n}}{q_{n^2}} = 
 \frac{\ln(1-(\frac{9}{10})^{\frac{1}{n}})}{\ln(1-(\frac{9}{10})^{\frac{1}{n^2}})}.
$$
The following limit is tedious to compute, but it can be shown that 
$$
 \frac{\ln(1-(\frac{9}{10})^{\frac{1}{n}})}{\ln(1-(\frac{9}{10})^{\frac{1}{n^2}})} \rightarrow \frac{1}{2}.
$$
Hence, we conclude that as $n\rightarrow \infty$ with probability at least $\frac{9}{200}$, $\left[\frac{X_{(n)}^{n^2}}{X_{(1)}^{n^2}} \right]\leq \frac{1}{2}+\epsilon$ for any $\epsilon>0$. In particular therefore,
$$\lim_{n\rightarrow \infty}\mE \left[\frac{X_{(n)}^{n^2}}{X_{(1)}^{n^2}} \right]<1.$$
Letting $p=1$,
$$
\lim_{n\rightarrow \infty} \EE[\delta_{\max}(v)]  = 1-
\lim_{n \rightarrow \infty} \mE\left[\frac{X_{(n)}^{n^2}}{X_{(1)}^{n^2}} \right]>0.
$$
Thus, even $k$ as small as $\sqrt{n}$ makes an MHR distribution have room to bid manipulation.
\end{proof}

\section{Extension to Deviations from BNE}

%\johcomm{Instead of asking whether truthfulness is an approximate equilibrium, it can also be of interest to study how close equilibrium is to truthfulness. That is, we can use the same metric and apply it to a BNE.}

\textcolor{teal}{We conclude by arguing that all our qualitative insights extend to the case where bidders consider deviating from equilibrium.}

\textcolor{teal}{A first-price auction with $n$ identical bidders and $k$ identical goods has a unique, symmetric Bayes-Nash equilibrium (BNE) $b(v)$ in which bidders shade their bids \cite{ChawlaH13}.}  %in which bidders shade their bid 

\textcolor{teal}{Then we can ask:}

\textcolor{teal}{1) How much can bidders change prices relative to what they would have paid if they were to bid truthfully?}

\textcolor{teal}{2) How much can bidders change prices relative to what they would have paid if they were to bid as in the BNE?}

%Beyond this, one could ask the following: instead of comparing $i$ bidding truthfully and $i$ bidding the highest other BNE bid, how about comparing $i$ bidding their BNE bid and $i$ bidding the highest other BNE bid. This is a reasonable question too, and this would amount to computing the metric of $E[\frac{b(X_2)}{b(X_1)}]$. By essentially the same argument as above, it follows that as $n$ goes to infinity, $E[\frac{b(X_2)}{b(X_1)}]$ will be the same as $E[\frac{X_2}{X_1}]$. So essentially all three of $E[\frac{b(X_2)}{X_1}]$, $E[\frac{b(X_2)}{b(X_1)}]$ and $E[\frac{X_2}{X_1}]$ coincide as $n$ goes to infinity.

\textcolor{teal}{In the first case, we would need to evaluate $1-\EE[b(X^n_{(k)})/X^n_{(1)}]$. In the second case, it would be $1-\EE[b(X^n_{(k)})/b(X^n_{(1)})]$.}
\textcolor{teal}{We show that the asymptotic behavior of both these quantities is identical to that of $1-\EE[X^n_{(k)}/X^n_{(1)}]$.}

%\johcomm{In other words, when every other bidder except $i$ is bidding as in a BNE, how do the auction-clearing prices for $i$ compare when (a) $i$ bids truthfully versus (b) $i$ bids the highest BNE bid from other bidders whenever $i$'s value exceeds the highest other BNE bid.}
\johcomm{
\begin{proposition}\label{BNE}
Let $b(v)$ be the unique Bayes-Nash Equilibrium for an agent with valuation $v$. Then for any $k\in [n]$, \begin{align*}
    &\text{1) \quad} \lim_{n\rightarrow \infty}\mE \left[\frac{b\left(X_{(k)}^n\right)}{X_{(1)}^n} \right] = \lim_{n\rightarrow \infty} \mE\left[\frac{X_{(k)}^n}{X_{(1)}^n} \right], \text{ and}\\
    &\text{2) \quad} \lim_{n\rightarrow \infty}\mE \left[\frac{b\left(X_{(k)}^n\right)}{b(X_{(1)}^n)} \right] = \lim_{n\rightarrow \infty} \mE\left[\frac{X_{(k)}^n}{X_{(1)}^n} \right].
\end{align*}
\end{proposition}}

That is, asymptotic price manipulability stays the same when all other agents bid according to the BNE, as when they are truthful. The main part of the argument is that by the revenue-equivalence theorem, the difference between a truthful bid and a BNE bid can be shown to have an upper bound $\beta^n$ for some $\beta \in (0,1)$. Since we take the limit $n \Rightarrow \infty$, the result follows.

\begin{proof}
\textcolor{teal}{We only provide a proof of 1); the argument for 2) is similar. We first consider the $k = 1$ case.}

The revenue equivalence theorem can be applied to derive the BNE bids. The revenue equivalence theorem guarantees that not just expected revenue are equal between the first- and second-price auctions, but also the payments made by any individual bidder at any given value. In a first price auction, a bidder with value $v$ pays $b(v)x(v)$ in expectation over other bidders' bids, where $b(v)$ is their BNE bid and $x(v)$ is the probability that they get allocated. We drop the subscript $i$ because we are in an i.i.d. setting. The second-price auction's payment formula is well known, namely $vx(v) - \int_{0}^{v} x(z) dz$. Equating this payment formula to $b(v)x(v)$, and dividing throughout by $x(v)$, we get,
\[
b(v) = v - \frac{\int_{0}^{v} x(z) dz}{x(v)}.
\]

Observe that the probability that a bidder with value $z$ wins in a first- or second-price auction is exactly the probability that every other bidder has a smaller value, namely $F(z)^{n-1}$. So the above can be rewritten as:
\[
b(v) = v - \frac{\int_0^v F(z)^{n-1}dz}{F(v)^{n-1}}.
\]

Our price-manipulability metric will now be $E[\frac{b(X_{(2)})}{X_{(1)}}]$ instead of what we currently have in the paper of $E[\frac{X_{(2)}}{X_{(1)}}]$, where $X_{(1)}$ and $X_{(2)}$ are random variables representing the largest and second-largest of $n$ i.i.d.~random variables. We now argue that as $n$ grows large $E[\frac{b(X_{(2)})}{X_{(1)}}]$ and $E[\frac{X_{(2)}}{X_{(1)}}]$ are \emph{identical}. To show this, let us consider the difference of the two metrics, namely, $E[\frac{X_{(2)}}{X_{(1)}}] - E[\frac{b(X_{(2)})}{X_{(1)}}]$. By the BNE bid formula above, this difference boils down to 
\[E\left[\frac{\frac{\int_0^{x_2} F(z)^{n-1}dz}{F(x_2)^{n-1}}}{x_1}\right],
\]
where $x_1$ is drawn from the distribution of $X_{(1)}$ and $x_2$ is drawn from the distribution of $X_{(2)}$. We are interested in the limit as $n$ goes to infinity of the above, namely,

$$\lim_{n\to\infty} E\left[\frac{\frac{\int_0^{x_2} F(z)^{n-1}dz}{F(x_2)^{n-1}}}{x_1}\right].$$

Notice that for any continuous distribution, regardless of whether it is regular or irregular, regardless of whether it is bounded or not bounded (which in particular includes every distribution we talk about like exponential, power-law, normal etc.), the quantity $F(z)/F(v) < 1$ for any $z < v$. Therefore, for any given $\epsilon > 0$, there exists $n(\epsilon)$ s.t., for all $n > n(\epsilon)$, we have $\frac{F(z)^{n-1}}{F(v)^{n-1}} < \epsilon$. Thus, we have:

$$\lim_{n\to\infty} E\left[\frac{\frac{\int_0^{x_2} F(z)^{n-1}dz}{F(x_2)^{n-1}}}{x_1}\right] <  \lim_{n\to\infty} E\left[\frac{\int_0^{x_2} \epsilon dz}{x_1}\right] = \lim_{n\to\infty} \epsilon \cdot E\left[\frac{x_2}{x_1}\right].$$

Since $\epsilon$ can be made arbitrarily small as $n$ grows larger, $\epsilon \cdot E[\frac{x_2}{x_1}]$ approaches $0$, completing our proof (notice that $E[\frac{x_2}{x_1}] \leq1$ regardless of $n$).

Notice that this proof remains unchanged even if we consider the $k$ unit auction (instead of the $1$-unit auction discussed above), and therefore compute the ratio of $E[\frac{b(X_{(k+1)})}{X_{(1)}}]$. By the exact same argument as above, we get that the new metric $E[\frac{b(X_{(k+1)})}{X_{(1)}}]$ will be equal to the current metric in the paper of $E[\frac{X_{(k+1)}}{X_{(1)}}]$.
\end{proof}

\section{Conclusion} 
% In this work we have initiated the exploration of non-truthful auctions through a metric of (non-)incentive compatibility that is based on the extent to which prices are manipulable. Using this metric we gave an almost complete picture of the manipulablity of large first-price auctions, and exhibited some surprising boundaries.

In this work, we adopt a metric from the bitcoin fee design market, that we call relative price change, to evaluate price manipulability in the non-truthful yet frequently used First-Price Auction. Using this metric we give an almost complete picture of the manipulability of large first-price auctions, and exhibit some surprising boundaries. We believe that the same metric could yield further insights when applied to other non-truthful auctions such as the Generalized Second-Price Auction. Moreover, due to its simplicity and practitioner friendliness coming out of its reliance only on bids and prices rather than true values, it would be feasible and instructive to evaluate the metric on real data.
%\newpage

%%
%% The next two lines define the bibliography style to be used, and
%% the bibliography file.
%\setcitestyle{numbers}
\bibliographystyle{siam}
\bibliography{sample-base}
%\balance

%%
%% If your work has an appendix, this is the place to put it.

% \newpage
% ~
% \newpage

\appendix

\clearpage

\newpage
\section{Appendix}

\subsection{Proof of Lemma~\ref{binomial_bound}}\label{app:binomial_bound}
\begin{lemma}[\textbf{Restatement of Lemma~\ref{binomial_bound}}]
For $m\sim B(n,p)$, and for any $i = o(np)$, $\lim_{n\rightarrow \infty}\Pr[m > i]=1$.
\end{lemma}

As mentioned earlier, while it is possible that Lemma~\ref{binomial_bound} is already known, we were unable to find a proof. So we give a proof here.

A strong tail bound on the binomial distribution when $i\leq np$ is given by

\begin{equation}\label{binomial_ineq}
\Pr[m \leq i] \leq \exp\left(-nD\left(\frac{i}{n}|| p\right)\right)
\end{equation}
Here $D(a||p)$ is the KL-divergence, defined as $D(a||p) = a\log(\frac{a}{p})+(1-a)\log(\frac{1-a}{1-p})$. 
%It is known that (\ref{binomial_ineq}) is nearly asymptotically tight ~\citep{Arratia89}.

Showing $\lim_{n\to\infty} nD\left(\frac{i}{n} || p\right) = \infty$ immediately proves Lemma~\ref{binomial_bound}. We show that $\lim_{n\to\infty} nD\left(\frac{i}{n} || p\right) = \infty$ in Lemma~\ref{tight_KL_bound}, which in turn uses Lemma~\ref{lower_fg_bound}. 

Since Lemma~\ref{binomial_bound} is for any $i=o(np)$, we begin by capturing $np$ in some functional form. Let $f(n) = np$, and let $i = g(np) = g(f(n))$, where $g(x)$ is $o(x)$, and $f(x)<x$, and $f(\cdot)$ is $\omega(1)$. Note that this is without loss of generality, and proving 
Lemma~\ref{binomial_bound} for such $i$ and $p$ is sufficient. This is because assuming that $f(n) (=np) < n$ just avoids the case of $p=1$, for which the proof of 
Lemma~\ref{binomial_bound} is trivial as things are deterministic. Also, assuming $f(.) = \omega(1)$, i.e., $np = \omega(1)$, is without loss of generality because if $np = O(1)$, there can be no $i$ s.t. $i = o(np)$. Also, since $i=o(np)$ and we have let $i = g(np)$, it immediately follows that $g(x) = o(x)$.

$$ D(\frac{i}{n}||p) = \frac{g(f(n))}{n}\log\left(\frac{g(f(n))}{f(n)}\right)+\left(1-\frac{g(f(n))}{n}\right)\log\left(\frac{1-\frac{g(f(n))}{n}}{1-\frac{f(n)}{n}}\right)$$

We begin by simplifying this expression, in particular the second term.

\begin{lemma}\label{lower_fg_bound}
Let $f(x)$ and $g(x)$ be non-negative, $f(x)< x$ and $g(x)$ is $o(x)$. $c<1$ a non-negative constant. Then for all $n$ sufficiently large,
$$\log\left(\frac{1-\frac{g(f(n))}{n})}{1-\frac{f(n)}{n}}\right) \geq c \cdot \frac{f(n)}{n}
$$
\end{lemma}

\begin{proof}
Our goal is to show that $\log\left(\frac{n- g(f(n))}{n-f(n)}\right) \geq c \cdot \frac{f(n)}{n}$.
%$$\log\left(\frac{1-\frac{g(f(n))}{n})}{1-\frac{f(n)}{n}}\right) = \log\left(\frac{n- g(f(n))}{n-f(n)}\right) $$
Apply the monotone increasing function $x \mapsto 2^x$ to both sides of the to-be-proven inequality. We now need to prove that:
$$
\Leftrightarrow \left(\frac{n- g(f(n))}{n-f(n)}\right) \geq 2^{(c \frac{f(n)}{n})}
$$
$$
\Leftrightarrow n \geq n2^{(c \frac{f(n)}{n})} - f(n)2^{(c \frac{f(n)}{n})} +g(f(n))
$$
Since $g(x)$ is $o(x)$, the following is a stronger inequality to prove:
$$
n \geq n2^{(c \frac{f(n)}{n})} - c\cdot f(n)2^{(c \frac{f(n)}{n})}
$$
$$
\Leftrightarrow (2^{(c \frac{f(n)}{n})} -1) \leq   c\cdot \frac{f(n)}{n}2^{(c \frac{f(n)}{n})}
$$
Set $z = 2^{(c \frac{f(n)}{n})}$.
Then above inequality becomes
$$z -1 \leq z\log z$$
where $z\geq 1$.
This holds because $z-1 = z\log z$ for $z=1$ and clearly $z\log z$ has larger derivative than $z$: $\frac{d}{dz}z\log z = \frac{1}{\ln(2)}+\log(z)>1$ for $z \geq 1$. 
\end{proof}

%Don't need this one at all and it might be wrong
%For completeness, we mention that Lemma \ref{lower_fg_bound} is tight.

%\begin{observation}\label{upper_fg_bound}
%Let $f(x)$ and $g(x)$ be non-negative monotone increasing and  $<x$. Then for $n$ sufficiently large,
%$$\log\left(\frac{1-\frac{g(f(n))}{n})}{1-\frac{f(n)}{n}}\right)\leq \frac{f(n)}{n}$$
%\end{observation}
%\begin{proof}
%$$\log\left(\frac{1-\frac{g(f(n))}{n})}{1-\frac{f(n)}{n}}\right)= \log\left(\frac{n- g(f(n))}{n-f(n)}\right)$$

%Apply the monotone increasing function $x \mapsto 2^x$ to both sides of the inequality.
%$$\Leftrightarrow \left(\frac{n- g(f(n))}{n-f(n)}\right) \leq 2^{(\frac{f(n)}{n})}$$ Rearranging we get

%\begin{align*}
%\Leftrightarrow n &\leq n2^{( \frac{f(n)}{n})} - f(n)2^{( %\frac{f(n)}{n})} +g(f(n)) \\
%\Leftrightarrow  n(2^{( \frac{f(n)}{n})} -1) &\geq  f(n)2^{( %\frac{f(n)}{n})} -g(f(n))
%\end{align*}
%Add and subtract $f(n)$ on the r.h.s.
%\begin{align*}
%\Leftrightarrow n(2^{(\frac{f(n)}{n})} -1) &\geq  %f(n)(2^{(\frac{f(n)}{n})}-1) + f(n) -g(f(n))\\
%\Leftrightarrow n &\geq  f(n) + \frac{f(n) -g(f(n))}{(2^{( %\frac{f(n)}{n})} -1)}
%\end{align*}

%A stronger inequality would be $$n \geq 2f(n)$$
%which clearly holds for $n$ sufficiently large.
%\end{proof}

Lemma \ref{lower_fg_bound} suffices to give a precise description of when it holds that
$\lim_{n\rightarrow \infty} nD(\frac{i}{n}|| p) = \infty$.

\begin{lemma}\label{tight_KL_bound}
Let $p=\frac{f(n)}{n}$ and $i=g(f(n))$, where $f(x)$ and $g(x)$ are non-negative, $f(x)< x$ and $f(x)$ is $\omega(1)$.
Let $g(x) = \frac{x}{h(x)}$ for some function $h$ such that $\lim_{x\rightarrow \infty} h(x) = \infty$, then $\lim_{n\rightarrow \infty} nD(\frac{i}{n}|| p) = \infty$
\end{lemma}
\begin{proof}
 $$D(\frac{i}{n}||p) = \frac{g(f(n))}{n}\log\left(\frac{g(f(n))}{f(n)}\right)+\left(1-\frac{g(f(n))}{n}\right)\log\left(\frac{1-\frac{g(f(n))}{n}}{1-\frac{f(n)}{n}}\right)$$
 For $n$ sufficiently large, by Lemma~\ref{lower_fg_bound} this implies  
\begin{align*}
D(\frac{i}{n}||p) &\geq 
\frac{g(f(n))}{n}\log(\frac{g(f(n))}{f(n)})+\frac{9}{10}\frac{f(n)}{n}\\
\Leftrightarrow
nD(\frac{i}{n}||p) &\geq g(f(n))\left(-\log(f(n)) +\log(g(f(n)))\right) + \frac{9}{10}f(n)
\end{align*}

Using the fact that $g(x) = \frac{x}{h(x)}$, we get:
\begin{align*}
nD(\frac{i}{n}||p) &\geq \frac{f(n)}{h(f(n))}\left(-\log(f(n)) + \log\left(\frac{f(n)}{h(f(n))}\right)\right) + \frac{9}{10}f(n)\\
&=-f(n)\frac{\log(h(f(n)))}{h(f(n))} + \frac{9}{10}f(n) \rightarrow \infty.
\end{align*}
where the last equality follows from the fact that $\frac{\log(x)}{x} \leq \frac{8}{10}$ for all $x > 0$.
\end{proof}

%\begin{corollary}\label{binomial_bound}
%For $i = o(np)$, $\lim_{n\rightarrow \infty}\Pr[B(n,p)>i]=1$.
%\end{corollary}
%\begin{proof}
Lemma~\ref{binomial_bound} now follows by combining inequality (\ref{binomial_ineq}) with Lemma \ref{tight_KL_bound}.

\end{document}